\documentclass[11pt,letterpaper]{article}

\usepackage{tablefootnote}
\usepackage{comment}
\usepackage{authblk}
\usepackage[numbers,sort]{natbib}
\usepackage[top=80pt,bottom=80pt,left=88pt,right=86pt]{geometry}
\usepackage{amssymb}
\usepackage{graphicx}
\usepackage[ruled,vlined,linesnumbered]{algorithm2e}
\usepackage{microtype}
\usepackage{subfig}
\usepackage{amsmath}
\usepackage{amsthm}
\usepackage{enumerate}
\usepackage{wrapfig}
\usepackage{multirow}
\usepackage{nicefrac}
\usepackage[colorlinks,linkcolor=blue,filecolor=blue,citecolor=blue,urlcolor=blue,pdfstartview=FitH,pagebackref]{hyperref}
\usepackage[nameinlink]{cleveref}
\usepackage{authblk}
\usepackage[dvipsnames]{xcolor}

\newcommand{\Reals}{{\mathbb{R}}}            
\newcommand{\eps}{\varepsilon}               

\newcommand{\partition}{\mathsf{Prt}}

\usepackage{mathtools}
\DeclarePairedDelimiter\segment {\langle}{\rangle}
\newcommand{\Frd}{Fr\'echet distance}

\def\A{{\cal A}}

\def\C{{\cal C}}
\def\D{{\cal D}}

\def\G{{\cal G}}

\def\I{{\cal I}}

\def\L{{\cal L}}

\def\N{{\cal N}}
\def\O{{\cal O}}
\def\P{{\cal P}}

\def\R{{\cal R}}
\def\S{{\cal S}}

\def\V{{\cal V}}

\newcommand{\figref}[1]{Figure~\ref{#1}}

\newcommand{\lemref}[1]{Lemma~\ref{#1}}

\newcommand{\thmref}[1]{Theorem~\ref{#1}}
\newcommand{\secref}[1]{Section~\ref{#1}}
\newcommand{\subsecref}[1]{Section~\ref{#1}}
\newcommand{\algref}[1]{Algorithm~\ref{#1}}

\usepackage{url}

\newtheorem{theorem}{Theorem}
\newtheorem{lemma}[theorem]{Lemma}

\def\eps{{\varepsilon}}

\begin{document}
	\title{On Approximate Near-Neighbors Queries among Curves under the (Continuous) Fr\'echet Distance}
	\date{}
	\author{Majid Mirzanezhad\thanks{Supported by the National Science Foundation grant CCF-1637576}}
	\affil[]{Department of Computer Science, Tulane University, USA}
	\affil[]{\textit {mmirzane@tulane.edu}}
	
	\maketitle
	\begin{abstract}
	Approximate near-neighbors search (\textsc{ANNS}) is a long-studied problem in computational geometry. 
In this paper, we revisit the problem and propose the first data structure for curves under the (continuous) Fr\'echet distance in $\Reals^d$. Given a set $\P$ of $n$ curves of size at most $m$ each in $\Reals^d$, 
and a real fixed $\delta>0$, we aim to preprocess $\P$ into a data structure so that for any given query curve $Q$ of size $k$, we can efficiently report all curves in $\P$ whose Fr\'echet distances to $Q$ are at most $\delta$. In the case that $k$ is given in the preprocessing stage, for any $\eps>0$ we propose a deterministic data structure whose space is  $n \cdot O\big(\max\big\{\big(\frac{\sqrt{d}}{\eps}\big)^{kd}, \big(\frac{\D\sqrt{d}}{\eps^2}\big)^{kd}\big\}\big)$ that can answer \textsc{$(1+\eps)\delta$-ANNS} queries in $O(kd)$ query time, where $\D$ is the diameter of $\P$. Considering $k$ as part of the query slightly changes the space to $n \cdot O\big(\frac{1}{\eps}\big)^{md} $ with $O(kd)$ query time within an approximation factor of $5+\eps$. We  show that our generic data structure for ANNS can give an alternative treatment of the approximate subtrajectory range searching problem studied by de Berg et al.~\cite{bcg-ffq-13}. We also revisit the time-window data structure for spatial density maps in~\cite{bndooh-twdsspd-20}. Given $\theta>0$, and $n$ time-stamped points spread over $m$ regions in a map, for any query window $W$, we propose a data structure of size $O(n/\eps^2)$ and construction time $O((n+m)/\eps^2)$ that can approximately return the regions containing at least $\theta$ points whose times are within $W$ in $O(1)$ query time. 
	\end{abstract}

\section{Introduction}
	Nearest neighbor search is a classical problem with a long history in computer science that has many applications in different areas such as pattern recognition, computer vision, DNA sequencing, databases, GIS, etc~\cite{l-hdrknn-91,bs-rnaspsaknn-06,hs-dbsd-99,glw-mpstd-07,rkv-nnq-95}.  In the elementary version of this problem, a set of points in $\Reals^d$ is given and for any query point the aim is to find the nearest point, in the point set, to the given query point. 
	Classical data structures, for instance point location using Voronoi diagrams, can solve the problem in $O(n^{\lceil d/2\rceil}\log n)$ preprocessing time that cause the curse of dimensionality issue as $d$ increases. 
	While a linear search serves to solve the problem readily, numerous methods have been employed to design data structures overcoming the running time bottleneck using approximation~\cite{amnsw-oaanns-98, k-tannshd-97, im-annrcd-98, kor-esannhds-00, diim-lshspd-04}. 
	There are several practical algorithms that make use of conventional data structures  including KD-tree, R-tree, X-tree, SS-tree, SR-tree, VP-tree, metric-trees for when the dimension $d$ is low~\cite{b-mdbstas-75, samet-06,wsb-ssmhds-98}. 
	
	There is relatively much less known about approximate nearest neighbor search (\textsc{ANNS}) among curves. Broadly speaking, curve data can be derived from a sequence of time stamped locations, so called \emph{trajectories}. Trajectories can be seen as piecewise linear functions since part of the movement tracked by the GPS between every two consecutive locations can be linearly interpolated. This can also be seen as a polygonal curve in the literature of computational geometry. 
	This steers the research to the point where the idea of developing data structures for curves as more complicated geometric data compared to points arises. Recently, the nearest neighbor search among curves has received a noticeable amount of attention. 
	The very first inspiring steps towards stating the nearest neighbor problem between trajectories (sequences of time-stamped locations) were initiated by Langran \cite{langran-tgis-99} and Lorentzos \cite{lorentzos-fermrmgi-88} that are more focused on indexing spatiotemporal data in large databases. 
	Not surprisingly research had mainly focused on indexing databases so that basic queries concerning the data can be answered efficiently. 
	At the time, the most common queries considered in the literature were different variants of nearest neighbor queries and range searching queries. Considering the nearest neighbor problem for trajectories requires the use of an appropriate distance for trajectories. 
	One of the most popular distances between polygonal curves that has interested many researchers over the past two decades is the Fr\'echet distance~\cite{f-sqpcf-06}. 
	Fr\'echet distance is a very well-known metric and has applications in many areas, e.g., in morphing~\cite{eghmm-nsmbp-02}, movement analysis~\cite{glw-mpstd-07}, handwriting recognition~\cite{skb-fdbas-07} and protein structure alignment~\cite{jxz-pssad-08}. It is intuitively the minimum length of the leash that connects a man and a dog walking across the curves without going backward. Alt and Godau~\cite{ag-cfdb-95} were the first who computed the Fr\'echet distance between curves with total complexity $n$ and their algorithm runs in $O(n^2 \log n)$ time. 
	For the \emph{discrete} Fr\'{e}chet distance, which considers only distances between the vertices, Agarwal et al.~\cite{aaks-cdfds-14} gave an $O\big(n^2 \frac{\log\log n}{\log n}\big)$ time algorithm. 
	
	There are different complementary variants of the approximate ``near-neighbors'' search (\textsc{ANNS}) problem among curves under the Fr\'echet distance such as approximate subcurve range searching (\textsc{ASRS}) and approximate subcurve range counting (\textsc{ASRC}). Such problems are closely related and they might be as difficult as the \textsc{ANNS}.
	While such problems are studied in the community of computational geometry there might be different titles devoted to them in the literature. 
	In this paper, we primarily set out to \textsc{ANNS} problem and we propose the first data structure that approximately handles the \textsc{ANNS} queries under the continuous Fr\'echet distance. The data strucure can be extended to handle various types of queries under the discrete Fr\'echet distance as well as the \textsc{ASRS} queries. 
	We will illuminate the backgrounds and existing results on the aforementioned problems later in Section~\ref{subsec:relatedwork}. The related problems that are commonly considered under the Fr\'echet distance are as follows:
	\\
	
	\noindent {\bf Problem 1 (\textsc{$(1+\eps)\delta$-ANNS}):}
	We are given a set $\P = \{P_1,\cdots, P_n\}$ of polygonal curves in $\Reals^d$ of size at most $m$ each and a real value $\delta>0$, we preprocess $\P$ into a data structure so that for any query curve $Q$ of size $k$, if the data structure returns a curve $P\in \P$ then its Fr\'echet distance to $Q$ is at most $(1+\eps)\delta$, and otherwise it is greater than $\delta$, for any $\eps>0$. 
	\\
	
	\noindent {\bf Problem 2 (\textsc{$(1+\eps)\delta$-ASRS}):} 
	We are given a polygonal curve $P$ of size $n$ in $\Reals^d$ and a range parameter $\delta>0$, we preprocess $P$ into a data structure so that for any query curve $Q$ of size $k$, if the data structure \emph{reports} and \emph{counts}, respectively, all subcurves $P' \subseteq P$ then their Fr\'echet distances to $Q$ are at most $(1+\eps)\delta$ and greater than $\delta$, otherwise. 
	\\
	
	Aside from \textsc{ANNS} queries among curves, spatiotemporal queries have attracted attention from researchers in GIS. One of the related problems in this area is time-window data structure for spatial density maps (\textsc{TWD}) that we investigate in this paper as well. 
	\\
	
	\noindent {\bf Problem 3 (\textsc{$(1+\eps)$-TWD}):}
	We are given a map $M$ consisting of a set of subdivisions (regions) $\R = \{r_1, \cdots, r_m\}$ and a time-stamped point set $\S = \{ (s_1, t_1) , \cdots, (s_n, t_n)\}$, where each point $s_i \in \Reals^d$ appears at time $t_i \in \Reals^{+}$ on the map, and an integer $\theta>0$. The aim is to preprocess $\S$ and $\R$ into a data structure so that for any query time-window $W = [q_1,q_2]$,  those regions in $\S$ are returned that contain at least $\theta$ points whose times are within $W$. 
	\\

	\subsection{Related work} \label{subsec:relatedwork}
	The first known result on \textsc{ANNS} for sequences of points (curves) under the discrete Fr\'echet distance has been obtained by Indyk~\cite{i-anndfpm-02}. For any parameter $t>1$, the data structure in \cite{i-anndfpm-02} uses $O\big(m^2|U|\big)^{tm^{1/t}}n^{2t}$ space and answers queries under the discrete Fr\'echet distance in $O\big(m+\log n\big)^{O(t)}$ query time within the approximation factor of $O\big((\log m +\log \log n)^{t-1}\big )$. Here $U$ is the ground set of the metric space in which the images of the input curves are defined. This data structure for $t=1+o(1)$, provides a constant factor approximation. Until recently in 2017, there was not much consideration of the problem under the Fr\'echet distance. Then Driemel and Silvestri~\cite{ds-lshfc-17} developed a locality-sensitive hash under the Fr\'echet distance. Their data structure is randomized and of size $O(n\log n + nmd)$ and can answer queries in $O(kd\log n)$ time within an approximation factor of $O(k)$. If they increase the space usage and query time to $O(2^{4md}n)$ and $O(2^{4md}\log n)$, respectively, then they can obtain an $O(d^{1.5})$ approximation factor. Subsequently, Emiris and Psarros~\cite{ep-pemapqmc-18} proposed a randomized data structure with approximation factor of $(1+\eps)$. 
	
	 Recently, Driemel et al. \cite{dps-sdfsfq-19} and Filtser et al.~\cite{ffk-anncsed-19} concurrently obtained randomized data structures under the \emph{discrete Fr\'echet} distance that can answer queries in $O(kd)$ query time within a $(1+\eps)$ approximation factor. While their space and construction times are quite different, their ideas in handling queries are somewhat similar as they both use  \emph{rounding} to speed up the query algorithm. The data structure by Driemel et al.~\cite{dps-sdfsfq-19} is randomized and has $O(kd)$ query time and can be derandomized with a significantly increased query time of $O\big(d^{5/2}k^2\eps^{-1}(\log n+kd\log(kd/\eps))\big)$. The only requirement of their data structure is that $k$, the size of the query, has to be provided in the preprocessing stage. Filtser et al.~\cite{ffk-anncsed-19} show that under the same assumption one can construct a data structure, in $n\cdot O(\frac{1}{\eps})^{kd}$ size and $nm\cdot \big(O(d\log m)+O(\frac{1}{\eps})^{kd}\big)$ expected time. In their terminology, this variant of the setting, when $k$ is given in the preprocessing stage, is called \emph{asymmetric}. In the \emph{symmetric} setting, when $k$ is assumed to be part of the query, they provide a data structure of size $n\cdot O(\frac{1}{\eps})^{md}$ and $n\cdot O(\frac{1}{\eps})^{md}$  expected construction time with $O(md)$ query time.  In the asymmetric deterministic case their construction time and query time increase to $nm\log(\frac{n}{\eps})\cdot (O(d\log m)+O(\frac{1}{\eps})^{kd})$ and $O(kd\log(\frac{nkd}{\eps}))$, respectively. Also in the symmetric deterministic case their construction and query times are  $n\log(\frac{n}{\eps})\cdot O(\frac{1}{\eps})^{md}$ and $O(md\log(\frac{nmd}{\eps}))$, respectively.
	 However, in the symmetric case, they have prior knowledge about the query size beforehand. What they assume at this point is that $Q$ has size at most $m$, i.e., $k\ll m$ in order to make their data structure work. 
	The construction of their data structure is randomized, therefore it results in expected preprocessing time.  Driemel et al. \cite{dps-sdfsfq-19} only considered the case when $k$ is known as part of the preprocessing, i.e., the asymmetric case. Very recently, Driemel and Psarros~\cite{dp-2anntsfd-21} revisited the problem for \emph{time series curves} (in $\Reals^1$) and gave a data structure of size $n\cdot O(\frac{m}{k\eps})^k+O(nm)$ that has $O(k2^k)$ query time and the approximation factor of $(2+\eps)$. 
	An attempt to summarize the existing and our results on the \textsc{ANNS} can be found in Table~\ref{tab:ANN_Res}.
	
	\begin{table}[!t]
		\centering 
		
		\begin{tabular}[c]{|l|c|c|c|c|}
			\hline 
			Result & Approx. & Query time & Space & Assumption \\ 
			\hline \hline
			
			det.~\cite{i-anndfpm-02}& $O(1)$ & $O(m+\log n)^{O(1)}$   & {\small $O(m^2|U|)^{m^{1-o(1)}}\cdot O(n^{2-o(1)})$} & dF, sym.\\
			\hline
			
			rand.~\cite{ds-lshfc-17}& $O(k)$ & $O(kd\log n)$   & $O\big(n(md+\log n)\big)$ & {dF, sym.} \\
			\hline
			
			rand.~\cite{ds-lshfc-17}& $O(d^{1.5})$ & $O(2^{4md}\log n)$   & $O(2^{4md}n)$ &  {dF, sym.} \\
			\hline
			
			rand.~\cite{ep-pemapqmc-18}& $ 1+\eps $ & $O(dm^{(1+\frac{1}{\eps})} 2^{4m} \log n )$   & $\hat{O}(n)\cdot (2+\frac{d}{\log m})^\lambda$  & {dF, sym.} \\
			\hline
			
			rand.~\cite{ffk-anncsed-19}& $ 1+\eps $ & $O(kd)$   & $n\cdot O(\frac{1}{\eps})^{kd}$ & {dF, asym.} \\
			\hline
			
			rand.~\cite{ffk-anncsed-19}& $ 1+\eps $ & $O(md)$   & $n\cdot O(\frac{1}{\eps})^{md}$ & {dF, sym.} \\
			\hline
			
			rand.~\cite{dps-sdfsfq-19}& $ 1+\eps $ & $O(kd)$   & $n\cdot O(\frac{kd^{3/2}}{\eps})^{kd}$ & dF, asym.\\
			\hline
			
			det.~\cite{dps-sdfsfq-19}& $ 1+\eps $ & $O\Big(\frac{\log n+kd\log(kd/\eps)}{\eps d^{-5/2}k^{-2}}\Big)$  & $(\frac{nkd}{\eps}) \cdot O(\frac{kd^{3/2}}{\eps})^{kd}$ & dF, asym. \\
			\hline
			
			det.~\cite{ffk-anncsed-19}& $ 1+\eps $ & $O(kd\log(\frac{knd}{\eps}))$   & $n\cdot O(\frac{1}{\eps})^{kd}$ & {dF, asym.} \\
			\hline
			
			det.~\cite{ffk-anncsed-19}& $ 1+\eps $ & $O(md\log(\frac{nmd}{\eps}))$   & $n\cdot O(\frac{1}{\eps})^{md}$ & {dF, sym.} \\
			\hline
			
			det.~\cite{dp-2anntsfd-21}& $ 2+\eps $ & $O(k2^k)$   & $n\cdot O(\frac{m}{k\eps})^k+O(nm)$ & {F, asym. $\Reals^1$} \\ 
			
			\hline\hline 
			
			det. Thm~\ref{thm:DS}& $ 1+\eps $ & $O(kd)$  & $n\cdot O(\max\{(\frac{\sqrt{d}}{\eps}),(\frac{\D\sqrt{d}}{\eps^2})\}^{kd})$  & F, asym. $\Reals^d$\\
			\hline
			
			det. Thm~\ref{thm:exDS}& $ 5+\eps $ & $O(kd)$  & $n\cdot O\big(\frac{1}{\eps}\big)^{md}$  & dF, sym. \\
			\hline
		\end{tabular}
		\caption{Results on \textsc{ANNS} under the (discrete) Fr\'echet distance. Here, $\lambda= {O(m^{(1+\frac{1}{\eps})}\cdot d \log(\frac{1}{\eps}))}$, and $\hat{O}$ hides factors polynomial in $1/\eps$. (d)F stands for the (discrete) Fr\'echet distance, rand. and det. stand for randomized and deterministic data structures, respectively.}
		\label{tab:ANN_Res}
	\end{table}
	
Previous research on \textsc{ANNS} under the Fr\'echet distance has been closely tied to developing data structures for a single curve. 
Designing data structures for one curve only has been initiated by a few important works~\cite{dh-jydcf-13,bcg-ffq-13,gs-fqgt-15}, where the type of queries varies based on the application. Compared to the \textsc{ANNS}, data structures for the single curve variant might seem relatively basic, however it might help us to solve the \textsc{ANNS} for multiple curves more efficiently. 
In 2013, Driemel and Har-Peled~\cite{dh-jydcf-13} considered approximate Fr\'echet distance queries (\textsc{AFD}) and presented a near-linear size data structure for an input curve $P$ such that for any given query curve $Q$ with $k$ vertices, and a subcurve $P'\subseteq P$, an $O(1)$-approximation of the Fr\'{e}chet distance between $Q$ and $P'$ can be computed in $O(k^2 \log n \log(k \log n))$ time. 
Concurrently, de Berg et al.~\cite{bcg-ffq-13}  considered $\delta$-\textsc{ASRC}, a slightly different setting than \cite{dh-jydcf-13}, and they gave a near quadratic-size data structure that can `approximately' count the number of subcurves of $P$, in polylogarithmic query time, whose Fr\'{e}chet distance to a given query segment $Q$ is at most $\delta$. 
Subsequently, Gudmundsson and Smid~\cite{gs-fqgt-15} considered the $\delta$-\textsc{ASRS} problem for a special case where the input curve is $c$-packed. A curve is $c$-packed if for any ball of radius $r>0$, the length of the curve contained inside of the ball is at most $2cr$. They showed for any constant $\eps > 0$, a data structure of size $O((c+(1/\eps^{2d})\log^2 (1/\eps)) n)$ can be built in $O(((1/\eps^{2d})\log^2 (1/\eps)) n \log^2 n + cn\log n )$ time, such that for any polygonal query path $Q$ of size $k$ the query algorithm answers the \textsc{$3(1+\eps)\delta$-ASRS} in $O((c^2/\eps^4) k \log n)$ time. 
A while after, de Berg et al.~\cite{dmo-dsfqtd-17} and Gudmundsson et al.~\cite{gmmw-ffdbcle-19} considered the Fr\'echet distance (\textsc{FD}) and $\delta$-Fr\'echet distance decision ($\delta$-\textsc{FD}) queries, respectively, and achieved $O(\log^2 n)$ query time under some specific assumptions. 
A summary of the existing results and our result can be found in Table \ref{tab:ASRS_Res}. Aside from developing near-neighbors data structures, we set out to approximate time-window query data structure in spatial density maps (heatmaps) among points. Heatmaps are robust tools for visualizing the distribution and density of different attributes over different layers of the map in GIS. Recently, Bonerath et al.~\cite{bndooh-twdsspd-20} studied the time-window queries for spatial density maps and they obtained a data structure of size $O(n \log m + m)$ and preprocessing time $O(n(\log n + \log m)+m)$ that can answer the queries in $O(k\log m + \log n)$, where $n$ and $m$ are the number of points and regions, respectively, and $k$ is the number of subdivisions in $\R$ containing points that fulfill $W$.

\begin{table}[!t]
	\small
	\centering 
	\begin{tabular}[c]{|l|c|c|c|c|} 
		\hline 
		Data structure & Approx. &Query time & Space& Assumption \\
		\hline\hline
		$\delta$-\textsc{ASRC}~\cite{bcg-ffq-13} & 2+3$\sqrt{2}$ & $O(\frac{n}{\sqrt{s}}\log^{O(1)}n)$ & $O(s\log^{O(1)}n)$ & $Q$:long seg, $\Reals^2$ \\
		\hline
		\textsc{AFD}~\cite{dh-jydcf-13} & $O(|U|)$\tablefootnote{In the proof of Lemma 6.8 in \cite{dh-jydcf-13} it is shown that the approximation factor is $O(r)$, where $r$ is the Fr\'echet distance between $Q$ and $P$. Clearly $r$ could be $\Omega(1)$, however one can bound the domain $U$ of the space, i.e., the diameter where the curve lives, so that $r\leq |U|= O(1)$.} & $O(k^2\log n \log( k\log n))$ & $O(n\log n)$ & $\Reals^d$ \\
		\hline
		$\delta$-\textsc{ASRS}~\cite{gs-fqgt-15} & $3(1+\eps)$ & $O(\frac{c^2}{\eps^4} k \log n)$ & $O\Big(n\big(c+\frac{\log^2 (\frac{1}{\eps})}{\eps^{2d}}\big) \Big)$ & $P$:$c$-packed, $\Reals^d$ \\
		\hline
		\textsc{FD}~\cite{dmo-dsfqtd-17} & 1 & $O( \log^2 n)$ & $O(n^2\log^2 n)$ & $Q$:segment, $\Reals^2$ \\
		\hline
		$\delta$-\textsc{FD}~\cite{gmmw-ffdbcle-19} & 1 & $O(k\log^2 n)$ & $O(n \log n)$ & $Q$:long, $\Reals^2$ \\
		\hline \hline 		
		\textsc{ASRS},Thm~\ref{thm:ASRS} & $1+\eps$ & $O(k)$ &  ${\tiny n \cdot O(\max\{(\frac{1}{\eps}), (\frac{\D}{\eps^2})\}^{kd})}$
		& $\Reals^d$, $k$ is given \\
		\hline
	\end{tabular}
	
	\caption{Known results and our result on data structures for \textsc{$\delta$-ASRS}, \textsc{$\delta$-ASRC}, \textsc{$\delta$-FD}, \textsc{AFD}, \textsc{FD} under the continuous Fr\'echet. Here $s$ is a parameter where $n\leq s \leq n^2$, and $|U|$ is interpreted as the diameter of the metric space where the curve is defined. 
	}
	\label{tab:ASRS_Res}
\end{table}

	\subsection{Our contribution}
In this paper we present a generic data structure that can approximately handle different variants such as \textsc{ANNS} in both symmetric and asymmetric ways as well as the \textsc{ASRS} problem. 
In general, the advantages of our data structure are as follows: (1) it works under the continuous Fr\'echet distance in $\Reals^d$, (2)  our data structure is fully deterministic, and (3) it is simpler than the approaches proposed in~\cite{dps-sdfsfq-19,ffk-anncsed-19, dp-2anntsfd-21}.
Our data structure is fully deterministic, meaning that both our preprocessing and query algorithms run deterministically. In the asymmetric case,
in time $nmkd \cdot  O\big(\max\big\{\big(\frac{\sqrt{d}}{\eps}\big)^{kd}, \big(\frac{\D\sqrt{d}}{\eps^2}\big)^{kd}\big\}\big)$ we build a data structure of size $n\cdot  O\big(\max\big\{\big(\frac{\sqrt{d}}{\eps}\big)^{kd}, \big(\frac{\D\sqrt{d}}{\eps^2}\big)^{kd}\big\}\big)$ such that for any polygonal query curve $Q$ of size at most $k$, it can answer \textsc{$(1+\eps)\delta$-ANNS} queries under continuous Fr\'echet distances in $O(kd)$ time. In the symmetric case, when $k$ is part of the query, we propose a data structure of size $n\cdot O\big(\frac{1}{\eps}\big)^{md}$ and construction time $ nm \cdot O\big(\frac{1}{\eps}\big)^{md}$ that answers \textsc{$(5+\eps)\delta$-ANNS} queries in $O(kd)$ time under the discrete Fr\'echet distance only.

To the best of our knowledge, only the data structure in~\cite{dp-2anntsfd-21} handles the ANNS among curves under the \emph{continuous Fr\'echet distance} but solely in $\Reals^1$. As mentioned above, the space of our data structure depends on the diameter $\D$ of the input curves. This is consistent with the result in~\cite{dp-2anntsfd-21}; any data structure that achieves an approximation factor less than 2 and supports curves of arclength at most $\D$, uses space of size $\D^{\Omega(k)}$. 
It is not obvious at all how to avoid this when approximating the solution within $(1+\eps)$ factor. 
The main idea of our construction is to use a grid of limited side length and store the near neighbor curves with respect to all possible ways of placing a query curve $Q$ onto the grid cells. Once $Q$ is given all we need is to search for a sequence of grid cells that contain the vertices of $Q$. We then report the preprocessed curves that are near neighbors to those cells that are close enough to the vertices of $Q$. The main ingredient of our approach is somewhat similar to \cite{dps-sdfsfq-19} and \cite{ffk-anncsed-19}, although the main challenge is to preprocess the curves and using a bounded uniform grid by discarding invalid placements of the query curve in order to achieve a deterministic data structure with faster query time. 

In the \emph{symmetric deterministic} case, we give an alternative data structure which is somewhat a trade-off for \cite{ffk-anncsed-19} and much simpler than the existing ones, however, Filtser et al. consider a symmetric case in which the size of the query is bounded by $m$ as discussed earlier. In our symmetric setting, which is slightly different and more general, our data structure has no prior knowledge about the size of the query in the preprocessing stage and $k$ can be either greater or less than $m$. Therefore, our data structure is `flexible' (still working) with respect to the value of $k$ and this flexibility adds a constant multiplicative factor of $5+\eps$ to the approximation.
The data structure can be extended to handle queries for the \textsc{$(1+\eps)\delta$-ASRS} problem. As mentioned, de Berg et al.~\cite{bcg-ffq-13} consider the $\delta$-ASRC problem, where the query is a single segment while our data structure can answer \textsc{$(1+\eps)\delta$-ASRS} for polygonal query curves. It has size of $n\cdot O\big(\max\big\{\big(\frac{1}{\eps}\big)^{kd}, \big(\frac{\D}{\eps^2}\big)^{kd}\big\}\big)$ and construction time of $nk\cdot O\big(\max\big\{\big(\frac{1}{\eps}\big)^{kd}, \big(\frac{\D}{\eps^2}\big)^{kd}\big\}\big)$ that can answer such queries in $O(k)$ time within $(1+\eps)$ approximation factor where, 
$d$ is assumed to be constant similar to \cite{bcg-ffq-13}. 
Finally, we consider the \textsc{TWD} problem and we show that for any $\eps>0$, a data structure of size $O(m/\eps^2)$ can approximately, in $O(1)$ query time, return two solutions $S_1$ and $S_2$ w.r.t. $[(1+\eps)q_1,(1-\eps)q_2]$ and $[(1-\eps)q_1,(1+\eps)q_2]$, respectively, such that $S_1 \subseteq S^*\subseteq S_2$, where $S^*$ is a solution w.r.t. the query window $W=[q_1,q_2]$. 

For the sake of simplicity, throughout the paper, we remove the number of reported curves (samples) from the runtimes since the query algorithms for \textsc{ANNS}, \textsc{ASRS}, \textsc{TWD} are output sensitive.
	
\section{Preliminaries} \label{sec:preliminary}
Let ${P=\langle p_1,p_2,\cdots, p_m \rangle}$ be an input
polygonal curve with $m$ vertices.
We treat $P$ as a continuous
map $P:[1,m] \rightarrow \mathbb{R}^d$, where $P(i)=p_i$ for integer $i$,
and the $i$-th edge is linearly parametrized as $P(i + t) =
(1-t) p_i + tp_{i +1}$,
for any $0 < t < 1$. 
We 
denote the {\em segment}, i.e., the  straight line connecting two points $p$ and $q$ on $P$, by $\segment{pq}$.  
The \emph{\Frd} between two polygonal curves $P$ and $Q$, with $m$ and $k$ vertices, respectively, is: 
${\delta_F(P, Q)}= \inf_{(\sigma,\theta)} \max_{t\in [0,1]} \| P(\sigma(t))-Q(\theta(t))\|,$ 
where $\sigma$ and $\theta$ are continuous non-decreasing functions from $[0,1]$ to $[1,m]$ and $[1,k]$, respectively with $\sigma(0)=\theta(0)=1$, $\sigma(1)=m$ and $\theta(1)=k$.
Finally, the {\em discrete \Frd} ($\delta_{dF}(P,Q)$) is a variant where $\sigma$ and $\theta$ are discrete functions from $\{1, \ldots, l\}$ to $\{1, \ldots, m\}$ and $\{1, \ldots, k\}$ with the property that  $|\sigma(i)-\sigma(i+1)|\le 1$ and $|\theta(i)-\theta(i+1)|\le 1$ for any $i\in \{1,\cdots,l-1\}$.

Let $\P=\{P_1,\cdots, P_n\}$ be  a set of $n$ polygonal curves such that each curve $P\in \P$ consists of $m$ vertices in $\Reals^d$ and $U\subseteq \Reals^d$ be the ground set of the metric space where the curves in $\P$ are defined. Let $V(P)$ denote the vertex set of $P \in \P$. Then $V(\P)$ is the set of vertices over all curves belonging to $\P$, i.e., $V(\P)= \cup_{ P\in \P}V(P)$.
The diameter $\D$ of $\P$ is then induced by the farthest pair of points in $V(\P)$. 

Let ${\partition(U, \ell)}$ be a {\em partitioning} of $U$ into a set of disjoint cells (hypercubes) of side length $\ell$ that is induced by axis-parallel hyperplanes placed consecutively at distance $\ell$ along each axis. This induces a \emph{grid region} $\G$. Intersecting $\G$ with an axis-parallel hypercube of side length $\L$ gives us a \emph{bounded} grid region of side length $\L$ whose grid cells have side length $\ell$. Clearly $\N=(\lceil \L/\ell\rceil)^d$ is the number of cells in $\G$. For any point $u\in U$ and real $r>0$, the ball centered at $u$ of radius $r$ is denoted by $B(u, r)$.

\section{Asymmetric \textsc{ANNS} under the Continuous Fr\'echet Distance} \label{sec:asymDS}
In this section we consider the decision version of the \textsc{ANNS} problem for polygonal curves under the continuous Fr\'echet distance. Our data structure is deterministic and simple.
	
\subsection{The main sketch}

What needs to be understood is the fact that the entire difficulty of the problem boils down to grid usage. 
More precisely, in our design, we use the following perspective: if a bounded grid is placed properly to enclose the entire curves in $\P$, then a brute force way of preprocessing the curves is efficient, necessary and sufficient. This way we can use the `rounding' idea similar to \cite{dps-sdfsfq-19} and \cite{ffk-anncsed-19} to achieve the fastest query time. The whole idea is feasible only by approximating the diameter of $\P$ and constructing a  bounded grid whose side length is proportionally equal to either $\delta$ or the diameter.

We now describe the construction of our data structure.  While the complexity of the query (which is at most some positive integer value $k$) must be known in this stage, the idea is to compute all paths of size $k$ whose vertices are grid points and distances are at most $(1+\eps/2)\delta$ to every input curve. 
Those input curves that are validated to have the \Frd~at most $(1+\eps/2)\delta$ to such grid paths, are stored in a hashtable by their indices. Note that every grid point has a unique id. Each bucket of the hashtable is associated with one unique grid path whose id is simply the concatenations of grid points' ids with respect to the order in which the grid points are encountered along the grid path. These concatenated ids are also called the hash codes. We need these ids so that we can retrieve the bucket of interest whose grid path is closest to the query curve. We will describe this in our query algorithm in \secref{subsec:query}. 
We begin with the preprocessing algorithm below and then we show some of the technical details that help us to analyze the data structure. 

\subsection{Preprocessing algorithm}  \label{subsec:prep}

The input to this procedure is $\P$, $k$ and $\delta$. The steps are described in \algref{alg:prep}:

\begin{algorithm}[htbp]
	\DontPrintSemicolon
	\SetKwFunction{Preprocessing}{\textsc{PreprocessingAlgorithm}}
	\caption{The Preprocessing Algorithm}
	\label{alg:prep}
	
	\BlankLine 
	\Preprocessing{$\P,k,\delta,\eps$}:
	
	Approximate $\D$ greedily: start from an arbitrary vertex in $V(\P)$ and find the farthest point to it. Call this value $\D'$ \label{step:diam}
	
	\lIf{$\D' \leq \delta$}{set $\L= 4\delta$ and $\ell= \eps\delta/(2\sqrt{d})$} \label{step:smalldiam}
	\lElseIf{$\D' > \delta$}{set $\L= 4\delta\D'/\eps$ and $\ell= \eps\delta/(2\sqrt{d})$}\label{step:largediam}
	Build $\G$ of side length $\L'=2\L$ and grid cells of side length $\ell$ centered at an arbitrary vertex in $V(\P)$ \label{step:grid}
	
	\ForAll{\mbox{sequences of grid points (grid paths) $\C = \segment{c_1,\cdots,c_k}$ in $\G$}}{\label{step:sequence}
		\ForAll{$P \in \P$}
		{\label{step:curves}
			\If{$\delta_F(P,\C)\leq (1+\eps/2)\delta$}{\label{step:frechet}
				Store the index of $P$ into the bucket of id associated with  $\C$ \label{step:store}} 
		}
	}
	
	
\end{algorithm}
In line~\ref{step:diam}  we approximate the diameter $\D$. The classic 2-approximation algorithm for finding the diameter of $N$ points in $\Reals^d$, with running time $O(Nd)$, is to choose an arbitrary point and then return the maximum distance to another point. The diameter is no smaller than this value and no larger than twice this value (\lemref{lem:diam}). In lines~\ref{step:smalldiam}-\ref{step:largediam} we check whether the approximate diameter $\D'$ is larger than $\delta$ or not and we set the appropriate values to $\L$ and $\ell$ in each case that help us to approximate the  solution later. Suppose $\L$ is the minimum side length under which if $\G$ is centered appropriately, can contain the entire $V(\P)$. In line~\ref{step:grid} we double the side length $\L$ in order to fit the whole input curves in $\P$ within the grid space $\G$ of side length $\L'=2\L$. This is because the center of the grid, a vertex in $V(\P)$, might be a maximum point with respect to some axis in $\Reals^d$ and hence the grid of side length $\L$ may not be able to cover the entire $\P$ along the axis. Doubling the length fixes the problem. Further details on this are mentioned in the approximation section. Finally in line~\ref{step:sequence} we consider all grid paths $\C$ with $k$ grid points  and in line~\ref{step:curves}-\ref{step:store} we store the indices of those curves $P\in \P$ that are approximately close to $\C$ (line \ref{step:frechet}) in an appropriate bucket of the hashtable (line \ref{step:store}). This can be done by associating distinct integers as ids to the grid points in $\G$. Once a grid path $\C$ (a sequence of grid points) is selected (line \ref{step:sequence}) the id associated with $\C= \segment{c_1,\cdots,c_k}$ is simply the concatenation of the id associated with each $c_i$ for $1\leq i\leq k$.
Now we are ready to analyze our data structure below:
\\

\noindent{\bf Preprocessing time:} 
In line \ref{step:diam} of  \algref{alg:prep} we first approximate the diameter $\D$ of $nm$ points. As described earlier we can compute $\D'$ in $O(|V(P)|)= O(nmd)$ time. Lines~\ref{step:smalldiam}-\ref{step:largediam} together take constant time. Constructing $\G$ in line~\ref{step:grid} takes $O(\N)$ time, where:

\begin{center}
	$\N:=\big(\frac{\L'}{\ell}\big)^d= \big(\frac{2\L}{\ell}\big)^d\leq \max\Big\{\Big(\frac{16d^{1/2}}{\eps}\Big)^d, \Big(\frac{16\D'd^{1/2}}{\eps^2}\Big)^d\Big\},\mbox{for all $\eps>0.$}$ 
\end{center}

The expensive part of \algref{alg:prep} is between lines~\ref{step:sequence}-\ref{step:store}. There are: $O(2^{dk}\N^k)$ 
grid paths in $\G$ (line~\ref{step:sequence}). Deciding the Fr\'echet distance between two curves of size $k$ and $m$ takes $O(kmd)$ time~\cite{ag-cfdb-95}, therefore given $n$ curves in $\P$ (line~\ref{step:frechet}) of size at most $m$ we have
$O\big(nmkd \cdot 2^{kd}\N^{k} \big).$ 
Finally line \ref{step:store} only takes $O(1)$. Plugging $\N$ into the latter upper bound yields the following lemma:

\begin{lemma} \label{lem:prep}
	The preprocessing time is $nmkd \cdot O\Big(\max\big\{\big(\frac{\sqrt{d}}{\eps}\big)^{kd}, \big(\frac{\D'\sqrt{d}}{\eps^2}\big)^{kd}\big\}\Big)$.
\end{lemma}


\noindent{\bf Space:} 
The space required is only for storing the indices of curves in $\P$ whose Fr\'echet distances to all grid paths $\C$ are at most $(1+\eps/2)\delta$ (line~\ref{step:store}). Therefore the space is proportional to the number of grid paths times the number of indices stored in each bucket associated with each grid path id: 
$O\big(n2^{kd}\N^{k}\big).$
Plugging $\N$ into the latter upper bound yields the following lemma:

\begin{lemma}  \label{lem:space}
	For any $\eps>0$, the space required is $n \cdot O\Big(\max\big\{\big(\frac{\sqrt{d}}{\eps}\big)^{kd}, \big(\frac{\D'\sqrt{d}}{\eps^2}\big)^{kd}\big\}\Big)$.
\end{lemma}

\subsection{Query algorithm} \label{subsec:query}
	
In this section we present our query algorithm. Our query algorithm is very simple and makes use of the property of the grid we constructed together with the hashtable that consists of the near-neighbors of certain grid paths. The query algorithm is described as follows:
	\\
	\noindent{\bf Step (1):} Given $Q=\segment{q_1,\cdots,q_{k}}$ and $k$ as the input parameters to the algorithm, check if any vertex of $Q$ lies outside of $\G$. If so, terminate the algorithm. \label{step:far}
	\\
	\noindent{\bf Step (2):} Find a sequence of grid cells $\A=\segment{a_1,\cdots, a_{k}}$ in $\G$ where every vertex $q_i$ of $Q$ lies within $a_i$. \label{step:locating}  
	\\
	\noindent{\bf Step (3):} Find an arbitrary corner $c_i$ of $a_i$ for all $i \in \{1,\cdots, k\} $ and make a sequence of grid points $\C=\segment{c_1,\cdots,c_{k}}$ in $\G$.  \label{step:snapping}
	\\
	\noindent{\bf Step (4):} Return the indices of the curves that are stored in the bucket associated with the id of $\C$. \label{step:retrieving}
	\\
	
	\noindent We now analyze the performance of our query algorithm and its quality of approximation.
	\\
	
	\noindent{\bf Query time:} Step (1) checks whether the query curve is close enough to any of the curves in $\P$ or not. This can be done in $O(kd)$ time to check if all vertices of $Q$ lie within $\G$ or not. Step (2) computes all cells in $\G$ that contain the vertices of $Q$. This can be done by performing a binary search per vertex of $Q$ over grid cells while comparing every vertex's coordinates with halfplanes passing through the coordinate and then splitting the grid space into $2^d$ subgrid spaces and recurse into the appropriate subgrid that contains the vertex until we reach the single cell. This takes $O(kd\log \N)= \O\Big(kd\log\big(\max\{(\frac{\sqrt{d}}{\eps}),(\frac{\D\sqrt{d}}{\eps^2})\}^{kd}\big)\Big)= O\Big(k^2d^2\big(\log d+\log(\max\{(\frac{1}{\eps}),(\frac{\D}{\eps})\})\big)\Big)$.  \\
	Step (3) only takes $O(d)$ per vertex, hence $O(kd)$ time overall and 	Step (4) takes $O(s)$ where $s$ is the number of curves in the solution. As we mentioned at the very beginning we remove this from the runtime for more brevity. Therefore, the total runtime is dominated by Step (2) which is $O\Big(k^2d^2\big(\log d+\log(\max\{(\frac{1}{\eps}),(\frac{\D}{\eps})\})\big)\Big)$. 
	
	Note that our deterministic data structure in the query stage as it is now, already runs faster than the deterministic one proposed by Driemel and Psarros in~\cite{dps-sdfsfq-19}. But nevertheless we can use the \emph{`rounding'} idea similar to \cite{dps-sdfsfq-19,ffk-anncsed-19} in order to improve the query time to $O(kd)$ which is what both papers obtain in their randomized data structures.
	\\
	
	\noindent{\bf Improved query time:} The idea is to instead of performing a binary search and finding a sequence of cells containing the vertices of $Q$, use the coordinates of each vertex of $Q$ only to decide what grid point has the rounded coordinates of this vertex's coordinates. This can be done by scaling $\G$ properly in the preprocessing stage such that every cell attains unit side length, therefore every grid point would be of integer coordinates. We remember this scale and apply it to the coordinates of every vertex of $Q$ in the query stage later. Once a query curve is given, we scale the coordinates first and then round the coordinates of $Q$ to obtain a sequence of grid points. We then use the id of the obtained grid path to find the bucket in which the curves in $\P$ are stored. Clearly, the rounding takes $O(d)$ time per vertex in $Q$, so the running time for Step (2) and (3) is $O(kd)$ time overall. Step (4) takes $O(s)$ where $s$ is the number of curves in the solution. 

\begin{lemma}\label{lem:query}
	Given a query curve $Q$ of size at most $k$, the query algorithm runs in $O(kd)$ time.
\end{lemma}

\noindent{\bf Approximation:} In this part we examine the quality of approximation and show that the preprocessing and query algorithms presented in \subsecref{subsec:prep} and \subsecref{subsec:query} together solve the \textsc{$(1+\eps)\delta$-ANNS}. First we need to show some properties in the following technical lemmas:

\begin{lemma}[Approximate diameter]\label{lem:diam}
	Let $\D'$ be the approximate diameter described in \algref{alg:prep}. Then $\D/2\leq \D'\leq \D$.
\end{lemma}
	
\begin{proof}
	The upper bound is obvious. For the lower bound assume that the diameter $\D$ is attained between two vertices $x,y \in V(\P)$ and let $\D'$ be attained between a pair of vertices $p,q \in V(\P)$. By the triangle inequality we have $\|x-y\|\leq \|p-x\| + \|p-y\|$. Observe that $\|p-x\|$ and $\|p-y\|$ are smaller than $\|p-q\|$, otherwise either of them would be $\D'$. Therefore: $$\D=\|x-y\|\leq \|p-x\| + \|p-y\| \leq \|p-q\| + \|p-q\| = 2\|p-q\|= 2\D',$$ which completes the proof.
\end{proof}

\begin{lemma}[Distant query curves] \label{lem:far}
	Let $\G$ be a grid of side length $\L'=2\L$ centered at an arbitrary vertex in $V(\P)$, where $\L= \min\{4\delta,4\delta\D'/\eps\}$ with $0<\eps\leq \delta$. If there exists a vertex $q\in Q$ that lies outside of $\G$, then $\delta_F(P,Q)>\delta$ for all curves $P \in \P$.
\end{lemma}
\begin{proof}
	We prove this by showing that $\|q-p\|>\delta$ for all points $p\in P$ and $P \in \P$.  Let $U$ be the set of all points (not necessarily vertices) in $\cup_{i=1}^{n}P_i$. Let $x$ and $y$ be two points in $U$ under which $\D$ is attained. Note that by a simple argument one can show that $x$ and $y$ are necessarily two vertices in $U$, i.e., $x,y\in V(\P)$. Now let $\mathsf{C}$ be the smallest enclosing hypercube of $U$. We define $\mathsf{C}_{\delta}$ to be the hypercube whose center is identical to the center of $\mathsf{C}$ and its sides are enlarged additively by an amount of $\delta$ along all axes in $\Reals^d$. Clearly, for any point $q'\in \Reals^d$ outside of $\mathsf{C_{\delta}}$ it holds that $\|q'-p\|>\delta$ for all points $p \in U$. Let $\mathsf{L}$ be the side length of $\mathsf{C_{\delta}}$. We thus have $\mathsf{L}\leq \D+2\delta$ because of the enlarging argument mentioned above (see \figref{fig:diam}). There are two cases to consider as follows:
	
	\begin{enumerate}
		\item $\D'\leq \delta$, hence $\mathsf{L}\leq \D+2\delta\leq 2\D'+2\delta \leq 4\delta$, 
		since $\D \leq 2\D'$ by \lemref{lem:diam}.
		
		\item $\D'>\delta$, hence $\mathsf{L}\leq \D+2\delta\leq 2\D'+2\D'\leq 4\D' \leq 4\D'\delta/\eps$, 
		Since $\D\leq 2\D'$ following \lemref{lem:diam} and $\delta\geq \eps$.
	\end{enumerate}
	Putting the two cases above together we have $\mathsf{L}=\min\{4\delta,4\delta\D'/\eps\}$. Now setting $\L=\mathsf{L}$ we get a grid space $\mathsf{G}$ that itself and its center are identical to $\mathsf{C_{\delta}}$ and the center of $\mathsf{C_{\delta}}$, respectively. Therefore $\mathsf{G}=\mathsf{C_{\delta}}$. 
	If the center of $\mathsf{G}$ changes over the vertices in $V(\P)$, then $\mathsf{G}$ may not cover the entire $U$. Therefore doubling its side length will yield a grid space $\G$ that covers the entire points in $U$. Therefore $\G$ has side length $\L'=2\L$ whose center can be an arbitrary vertex in $U$. Since $\mathsf{C_{\delta}} = \mathsf{G} \subseteq \G$ then for any point $q'\in \Reals^d$ outside of $\G$ it holds that $\|q'-p\|>\delta$  for all points $p \in U$. This implies that for any vertex $q \in Q$ outside of $\G$ we have $\|q-p\|>\delta$ as well which results in $\delta_F(P,Q)>\delta$. 
\end{proof}
\begin{figure} [!t]
	\begin{center}
		\includegraphics[width=0.3\textwidth]{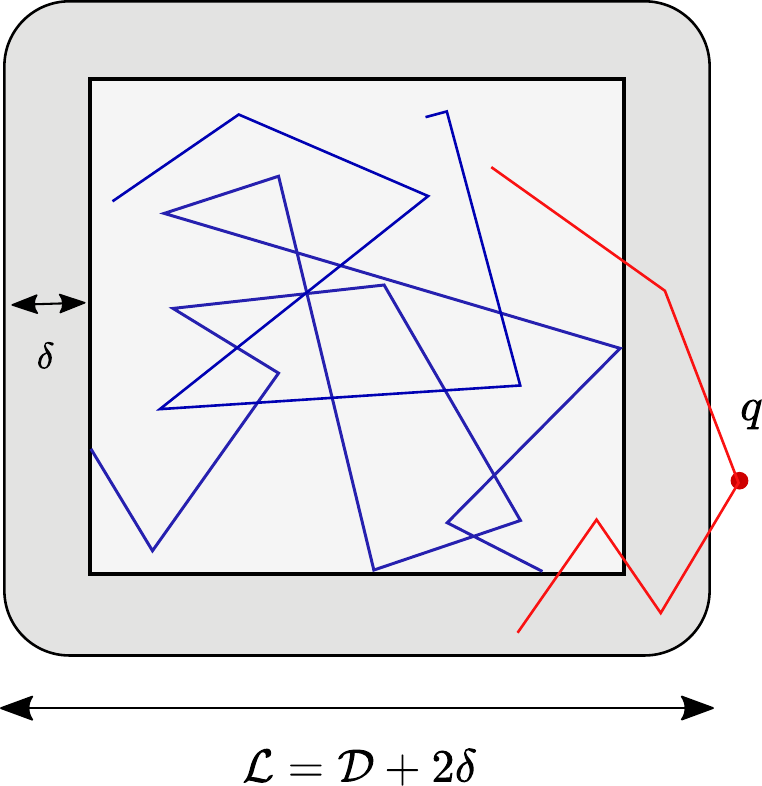}
	\end{center}
	\caption{\label{fig:diam} Any point $q$ outside the enlarged minimum enclosing cube of $\P$ cannot have close distance to any point inside of the cube.}
\end{figure}

\begin{lemma}[Close grid query path]\label{lem:closegridpath}
	Let $Q= \segment{q_1,\cdots, q_k}$ be a query curve in $\Reals^d$ and let $\A=\segment{a_1,\cdots,a_k}$ be the sequence of cells of side length $\ell=\eps\delta/(2\sqrt{d})$, where each $a_i$ contains $q_i$ for all $1\leq i\leq k$. For any grid path $\C=\segment{c_1,\cdots,c_k}$, where every $c_i$ is an arbitrary corner of $a_i$, it holds that $\delta_F(Q,\C)\leq \eps\delta/2$.
\end{lemma}
\begin{proof}
	According to \algref{alg:prep} the side length of each grid cell is equal to:
	$\ell=\eps\delta/2\sqrt{d}$, either $\D'\leq \delta$ or $\D' >\delta$ . Let $\mathsf{D}$ denote the diameter of each grid cell $a_i$. Then $\mathsf{D}=\ell_1\cdot\sqrt{d} =(\eps\delta/(2\sqrt{d}))\cdot\sqrt{d}  =\eps\delta/2$. 
	
	This implies that $\|q_i-c_i\|\leq \mathsf{D} \leq \eps\delta/2$ for every arbitrary corner $c_i \in a_i$. A simple Fr\'echet matching of width at most $\eps\delta/2$ (linearly) matches every $q_i$ to the corresponding $c_i$ and therefore $\delta_F(Q,C)\leq \eps\delta/2$.
\end{proof}

\begin{lemma}[Approximation]\label{lem:approx}
	Let $\eps>0$, and $P \in \P$ be a polygonal curve. If $P$ is returned by the query algorithm then $\delta_F(P,Q)\leq(1+\eps) \delta$. If $P$ is not returned then $\delta_F(P,Q)> \delta$. 
\end{lemma}
\begin{proof}
	Let $\A=\segment{a_1,\cdots,a_k}$ be the sequence of cells containing the vertices of $Q$ and $\C=\segment{c_1,\cdots,c_k}$ be a grid path where every $c_i$ is an arbitrary corner of $a_i$.
	Following~\lemref{lem:closegridpath} $\delta_F(Q,\C)\leq \eps\delta/2$.
	Since $P$ is returned, it is already stored in the hashtable in line~\ref{step:store} of~\algref{alg:prep} and $\delta_F(P,\C)\leq (1+\eps/2)\delta$. Applying a triangle inequality between $P$, $\C$ and $Q$ yields: $\delta_F(P,Q)\leq \delta_F(Q,\C) + \delta_F(P,\C) \leq \eps\delta/2 + (1+\eps/2)\delta = (1+\eps)\delta.$
	
	Now suppose $P$ is not returned by the query algorithm. Since $P$ is not returned it is not stored in the bucket of $\C$'s id in line~\ref{step:store} of~\algref{alg:prep} hence $\delta_F(P,\C)> (1+\eps/2)\delta$. Applying the other side of the triangle inequality yields: $\delta_F(P,Q)\geq |\delta_F(P,\C)- \delta_F(Q,\C)|> (1+\eps/2)\delta - \eps\delta/2 = \delta.$
	This completes the proof.
\end{proof}
	We summarize this section with the following theorem:
	\begin{theorem}\label{thm:DS}
	Let $\P=\{P_1,\cdots P_n\}$ be a set of $n$ polygonal curves in $\Reals^d$ each of size at most $m$, $\delta>0$ be a real number and $k$ be the size of query given beforehand. For any $0<\eps \leq \delta$, one can construct a deterministic data structure of size  $n \cdot O\Big(\max\big\{\big(\frac{\sqrt{d}}{\eps}\big)^{kd}, \big(\frac{\D\sqrt{d}}{\eps^2}\big)^{kd}\big\}\Big)$ and $nmkd \cdot O\Big(\max\big\{\big(\frac{\sqrt{d}}{\eps}\big)^{kd}, \big(\frac{\D\sqrt{d}}{\eps^2}\big)^{kd}\big\}\Big)$ preprocessing time, such that for any polygonal query curve $Q$ of size $k$ it computes the \textsc{$(1+\eps)\delta$-ANNS} under the continuous Fr\'echet distance in $O(kd)$ query time.  
	\end{theorem}
	\begin{proof}
	The construction time, space and query time follow from Lemmas \ref{lem:prep}, \ref{lem:space} and \ref{lem:query}, respectively,  where 
$\D' \leq \D$ by \lemref{lem:diam}. In Step (1) of the query algorithm if any vertex of $Q$ is outside of $\G$ of side length $2\cdot\max\{4\delta,4\delta\D'/\eps\}$ with $\D/2\leq \D'\leq \D$ (\lemref{lem:diam}), then following \lemref{lem:far} we have $\delta_F(P,Q)>\delta$. Hence the algorithm terminates as desired because there is no $P$ close to $Q$. In this lemma it is assumed that $\delta\geq \eps$.
Now in Step (2) and Step (3) the algorithm finds a grid path $\C$ of distance at most $\eps\delta/2$ to $Q$ (\lemref{lem:closegridpath}). \lemref{lem:approx} implies that if any curve like $P\in \P$ is returned by the query algorithm from the bucket of the id associated with $\C$ then $\delta_F(P,Q)\leq (1+\eps)\delta$ and $\delta_F(P,Q)>\delta$, otherwise. 
Note that all Lemmas~\ref{lem:far},~\ref{lem:closegridpath} and \ref{lem:approx} work under the discrete Fr\'echet distance as well, therefore they result in solving the $(1+\eps)\delta$-\textsc{ANNS} under the discrete Fr\'echet distance.
	\end{proof}

\section{Symmetric \textsc{ANNS} under the Discrete Fr\'echet Distance} \label{sec:symDS} 
In this section we set out to extending the data structure for the case where $k$ is not part of the preprocessing but part of the query. Since $k$ is not given in the preprocessing stage, the method presented in Section~\ref{sec:asymDS} may not be efficient for computing all grid path of size at most $k$. which leads to exponential growth in both space and construction time complexities. As mentioned in the previous section,  the way of building and using the grid is crucial for handling queries. The idea is instead of constructing a bounded grid of a certain side length, discretize the possible space for placement of query vertices and then preprocess the grid paths of length $k$ induced by the obtained discretization. The rest of the preprocessing and query algorithm together is similar to what we presented in \secref{sec:asymDS}. 
What Filtser et al.~\cite{ffk-anncsed-19} consider is that both the query and input curves' complexities are equal to $m$. They distinguish between this setting and the one where the size of the query may not be equal to the input curve's but is instead given in the preprocessing. However, it seems that this may not really affect the nature of the problem since in either case they assume that the size of the query curve is bounded by some known amount (either $k$ or $m$). What we assume is more general, where we have no prior knowledge about the size of the query curve beforehand.
	We describe our preprocessing algorithm below. 
	\\
	
	\noindent{\bf Step (1):}  Compute grid point sets $g_{i,j}=\{p~|~ p \in B\big(p_j,3(1+\eps/2)\delta\big)\cap \partition(\Reals^d,\eps\delta/2\sqrt{d})\}$, for all $p_j \in P_i$ with $1\leq j\leq m$ and $1\leq i \leq n$. 
	\\
	
	\noindent{\bf Step (2):} Compute grid point sets $g'_{i,j}=\{p~|~ p \in B\big(p_j,(1+\eps/2)\delta\big)\cap \partition(\Reals^d,\eps\delta/2\sqrt{d})\}$, for all $p_j \in P_i$ with $1\leq j\leq m$ and $1\leq i \leq n$. The obtained set of grid points is  \emph{marked} grid points. These points are candidates for placing the vertices of the rounded query grid path.
	\\
	
	\noindent{\bf Step (3):} Construct a graph $G_i$ for each $P_i \in \P$ whose vertices are $g_{i,j}$ for all $1 \leq j <|V(P_i)|$ and edges $\segment{c_jc_{j+1}}$ are straight-line segments, where $c_j \in g_{i,j}$ and $c_{j+1}\in g_{i,j+1}$.  
	\\
	
	\noindent{\bf Step (4):} For each \emph{complete} grid path $C_i\in G_i$, i.e., starting from a grid point in $g_{i,1}$ and ending to $g_{i,l}$, where $l=|V(P_i)|\leq m$, associate a bucket with a unique id induced by the concatenation of its grid points' ids, and store the index of $P_i$ (say $i$) in the bucket. 
	\\
	
	\noindent We are now ready to analyze our data structure: 
	\\
	
	\noindent{\bf Construction time:} Let $\N'$ be the number of cell obtained by $B\big(p_j,3(1+\eps/2)\delta\big)\cap \partition(\Reals^d,\eps\delta/2\sqrt{d})$. First note that:
	$$\N':= \frac{\V^d(B)}{\V^d(a)}.$$
	In above formula $\V^d(B)$ is the volume of the ball $B\big(p_j,(1+\eps/2)\delta\big)$ in Euclidean space which is divided by the volume of the cell $\V^d(a)$ of side length $\ell =\eps\delta/2\sqrt{d}$. Thus:
	
	
	$$\N':= \frac{{\frac{2\pi^{d/2}}{d\Gamma(d/2)}}}{\ell^d}\cdot R^d = \frac{{\frac{2\pi^{d/2}}{d\Gamma(d/2)}}}{(\eps\delta/2\sqrt{d})^d}\cdot \big(3(1+\eps/2)\delta\big)^d,$$
	
	which is:
	
	$$\N':= O\Big(\frac{1}{\eps}\Big)^d,$$
	for sufficiently small $\eps>0$. The upper bound of $\N'$ is obtained by dividing the volume of a Euclidean ball by the volume of a cell using a volumetric argument that appeared in \cite{him-anntrcd-12}.
	\\
	
	Now realize that the number of grid points in $g_{i,j}$ is:  $$|g_{i,j}|= 2^d\N'= 2^d\cdot O\Big(\frac{1}{\eps}\Big)^d = O\Big(\frac{1}{\eps}\Big)^d.$$
	The number of vertices in $G_i$ is: $$|V(G_i)|= \sum_{j=1}^{m} |g_{i,j}| \leq m |g_{i,j}| = m\cdot O\Big(\frac{1}{\eps}\Big)^d,$$ for all $1\leq i \leq n$,
	and the number of edges of $g_{i,j}$ is: 
	$$|E(G_i)|= (m-1)\cdot |g_{i,j}|^2 = m\cdot O\Big(\frac{1}{\eps}\Big)^{2d},$$ since we connect an edge between any pair of grid points lying within balls around every two consecutive vertices of $P_i$. For each complete path in $G_i$ we associate a bucket in the hashtable to store $i$ in it in constant time. The number of complete paths per $G_i$ is: $$O(|g_{i,j}|^m) = O\Big(\frac{1}{\eps}\Big)^{md}.$$ 
	
	We associate a bucket to each complete path in $G_i$. Therefore the total construction time for constructing $G_i$ and storing all paths of $G_i$ in the hashtable is: $$O(|V(G_i)|+ |E(G_i)|) + O(|g_{i,j}|^m) = m\cdot O\Big(\frac{1}{\eps}\Big)^{2d} + O\Big(\frac{1}{\eps}\Big)^{md} = m\cdot O\Big(\frac{1}{\eps}\Big)^{md}.$$ We have the following lemma:	

\begin{lemma}\label{lem:asymdfprep}
	The construction time of the data structure is $nm\cdot O\big(\frac{1}{\eps}\big)^{md}.$ 
\end{lemma}
	
\noindent{\bf Space:} The number of complete paths of length $m$ in $G_i$ of the grid point set size $|g_{i,j}|$ is: $$ O\big(|g_{i,j}|^m \big) = O\Big(\frac{1}{\eps}\Big)^{md},$$ for some $1\leq j\leq m$ and for all $1\leq i\leq n$. Note that number of buckets per graph $G_i$ is the same as the number of paths in $G_i$. And for each grid path in $\C_i \in G_i$ we only store $i$ in the bucket associated with $\C_i$, for $1\leq i \leq n$. Therefore we have $n\cdot O\big(\frac{1}{\eps}\big)^{md}$ space overall.
	
	\begin{lemma}\label{lem:asymdfspace}
		The space required for the data structure is $n\cdot O\big(\frac{1}{\eps}\big)^{md}$.
	\end{lemma}
	\subsection{Query algorithm} 
Now we present the query algorithm which is similar to what we presented for the previous data structure. We first check if all vertices of $Q$ lie entirely inside of $G= \cup_{i=1}^{n} G_i$. 
Next we round $Q$ onto a grid path $\C_Q$ whose vertices are only selected from the marked grid points in the preprocessing stage. We then simplify the rounded query grid path using $\mu$-simplification by~\cite{dhw-cfdrinlt-12}. This results in a simplified query grid path $\C'_Q$ whose edges are longer than $\mu$. We choose $\mu=2(1+\eps/2)\delta$ as the threshold parameter for this simplification. In the end we use the id of $\C'_Q$ to retrieve all stored curves from the hashtable. The $\mu$-simplification and rounding take linear time each so the query algorithm takes $O(kd)$ overall.
We have the following lemmas:

\begin{lemma}[Algorithm 2.1 in \cite{dhw-cfdrinlt-12}]\label{lem:exsimplification}
	There is a linear time $\mu$-simplification of $P$ that gives $P'$ fulfilling $\delta_{dF}(P,P')\leq \mu$, where every edge of $P'$ has length strictly greater than $\mu$ except possibly the last one. 
\end{lemma}

	\begin{proof}
	We describe the algorithm for the sake of a comprehensive presentation. The simplification is as follows: set the first vertex of $P$ as the current vertex $u$. Start from $u$ and find the first vertex $v$ along $P$ that lies outside of $B(u,\mu)$. Draw a line segment between $u$ and $v$ as a simplified link and set $v$ as the current vertex. Repeat this process until the last vertex is reached. 
Note that for every edge $\segment{uv}\in P'$ all points in $P$ before reaching $v \in P'$ are inside $B(u,\mu)$ thus they all can be matched to $u$ and therefore $\delta_{dF}(P,P')\leq \mu$. Also since $v$ is outside the ball $B(u,\mu)$ the length of $\segment{uv}$ is greater than $\mu$. Clearly every vertex is processed only once so it all takes linear time.
	\end{proof}
\begin{lemma}\label{lem:exquery}
	For a query curve $Q$ of size $k$, the query algorithm takes $O(kd)$.
\end{lemma}
\begin{proof}
	Following \lemref{lem:exsimplification} and the fact that rounding takes linear time, the query algorithm takes time $O(kd)$.
\end{proof}
\begin{lemma} [Correctness] \label{lem:excorrectness}
	Let $\eps>0$  and $P_i \in \P$ be a polygonal curve for some $1\leq i\leq n$. If $P_i$ is returned by the query algorithm then $\delta_{dF}(P_i,Q)\leq (5+\eps) \delta$. If $P_i$ is not   returned then $\delta_{dF}(P_i,Q)> \delta$. 
\end{lemma}
	
\begin{proof}
Suppose $P_i$ is stored in the hashtable already and returned by the query algorithm.  
Note that $\delta_{dF}(P_i,\C_i)\leq 3(1+\eps/2)\delta$, for all $\C_i\in G_i$. Since $P_i$ is returned  there exists a $\C^*_i\in G_i$ such that $\C^*_i=\C'_Q$ because there is a bucket of id associated with $\C^*_i$ that $\C'_Q$ caught it through the query algorithm. Thus $\delta_{dF}(P_i,\C'_Q)\leq 3(1+\eps/2)\delta$. Note that  $\delta_{dF}(Q,\C_Q)\leq \eps\delta/2$ because the distance between every vertex and the rounded one is at most the diameter of the grid cells (\lemref{lem:closegridpath}). On the other hand $\delta_{dF}(\C'_Q,\C_Q)\leq 2(1+\eps/2)\delta$ because of the $2(1+\eps/2)\delta$-simplification between $\C_Q$ and $\C'_Q$ (\lemref{lem:exsimplification}). Now applying the triangle inequality twice, once between $P_i$, $Q$ and $\C'_Q$ and once again between $\C'_Q$, $\C_Q$ and $Q$, yields:

\begin{align*}
\delta_{dF}(P_i,Q) & \leq\delta_{dF}(P_i,\C'_Q) +\delta_{dF}(\C'_Q,Q) \\
& \leq \delta_{dF}(P_i,\C'_Q) +\delta_{dF}(\C'_Q,\C_Q) + \delta_{dF}(\C_Q, Q) \\
& \leq 3(1+\eps/2)\delta + 2(1+\eps/2)\delta + \eps\delta/2 = (5+3\eps)\delta,
\end{align*}
	\begin{figure} [!t]
	\begin{center}
		\includegraphics[width=0.4\textwidth]{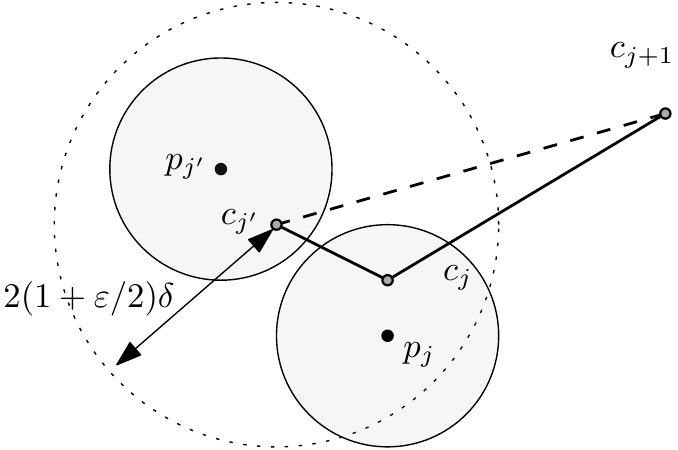}
	\end{center}
	\caption{\label{fig:proofCase2} The marked grid point $c_j \in B(p_j, (1+\eps/2)\delta)$ is removed on the solid line path by the simplified dashed line since $c_j \in B(c_{j'}, 2(1+\eps/2)\delta)$. }
\end{figure}
as desired when $\eps$ is sufficiently small. 
Now suppose that $P_i$ is not returned. We show that $\delta_F(P_i,Q)>\delta$. There are two cases occurring if $P_i$ is not returned: 

\begin{enumerate}
	\item It might be that all vertices of $\C_Q$ are not entirely rounded onto the marked grid points. Clearly this results in having $\delta_{dF}(P_i,\C_Q)>(1+\eps/2)\delta$ because some vertex of $\C_Q$ lies outside of $B\big(p_j, (1+\eps/2)\delta\big)$ for some $1\leq j\leq |V(P_i)|$. Now this time applying the other side of triangle inequality yields:
	$$\delta_{dF}(P_i,Q)> |\delta_{dF}(P_i,\C_Q)- \delta_{dF}(Q,\C_Q)| > (1+\eps/2)\delta - \eps\delta/2 = \delta .$$
	
	\item Case (1) does not occur but $\C'_Q \notin G_i$. For the sake of a contradiction assume that $\delta_{dF}(P,Q)\leq \delta$. Since Case (1) is not occurring it implies that all vertices of $\C_Q$ are rounded onto the marked grid points hence they are a subset of $V(G_i)$. Following \lemref{lem:exsimplification}, any $2(1+\eps/2)\delta$-simplification results in having simplified curve of edges longer than $2(1+\eps/2)\delta$ thus $\C'_Q$ has at most one vertex within every $B(p_j,(1+\eps/2)\delta)$ for all $1\leq j\leq |V(P_i)|$. Now if it has exactly one vertex inside of each such ball then clearly $\C'_Q \in G_i$. If there is a ball $B(p_j,(1+\eps/2)\delta)$  for some $1\leq j\leq |V(P_i)|$ such that no vertex of $\C'_Q$ falls into it, 
	there is at least a vertex $c_j\in \C_Q$ where $c_j \in B(p_j,(1+\eps/2)\delta)$ but it is removed due to the $2(1+\eps/2)\delta$-simplification. Therefore $c_j \in B(p_{j'},2(1+\eps/2)\delta)$ for some $j'<j$. This implies that there is some $c_{j'} \in \C_Q$ where $c_{j'} \in B(p_{j'},(1+\eps/2)\delta)$. Therefore: $$\|p_j-c_{j'} \| \leq \| p_j-c_j\|+ \|c_j-c_{j'}\|\leq (1+\eps/2)\delta + 2(1+\eps/2)\delta=3(1+\eps/2)\delta,$$
	
	as $\|c_j-c_{j'}\|\leq 2(1+\eps/2)\delta$ because of the fact that $c_j$ is removed by the simplification since  $c_j \in B(c_{j'}, 2(1+\eps/2)\delta)$, i.e., $c_j$ lies within $2(1+\eps/2)\delta$ distance from $c_{j'}$, see Figure~\ref{fig:proofCase2}.

	Since $\|p_j-c_{j'} \|\leq 3(1+\eps/2)\delta$, it follows that $c_{j'} \in g_{i,j}$ and correspondingly $\C'_Q\in G_i$ since $\C'_Q$ has now exactly one vertex in $B(p, 3(1+\eps/2)\delta)$ for all $p \in V(P_i)$. This implies that $P_i$ is stored in the bucket of $\C'_Q$'s id and $P_i$ is returned already. We have a contradiction. 
\end{enumerate} 
Therefore $\delta_{dF}(P,Q)>\delta$ in both cases and this completes the proof.
\end{proof}

	\noindent We now summarize the section with the following theorem:
	
	\begin{theorem} \label{thm:exDS}
	Let $\P=\{P_1,\cdots P_n\}$ be a set of $n$ polygonal curves in $\Reals^d$ each of size at most $m$, $\delta > 0$ be a real number. For any $\eps>0$, one can construct a deterministic data structure of size $n\cdot O\Big(\frac{1}{\eps}\Big)^{md}$ and construction time $nm\cdot O\Big(\frac{1}{\eps}\Big)^{md}$ such that for any polygonal query curve $Q$ of size $k$, it computes the \textsc{$(5+\eps)\delta$-ANN} under the discrete Fr\'echet distance in $O(kd)$ query time.
	\end{theorem}

	\section{Approximate subtrajectory range searching queries} \label{sec:ASRS}
	
Another application of our data structure is the capability of solving the approximate subtrajectory range searching (\textsc{$\delta$-ASRS}) problem for when $k$ is small. 
De Berg et al. \cite{bcg-ffq-13} gave an initiatial treatment of the  $\delta$-\textsc{ASRC} problem for the case that $Q$ is only a single line segment and $P$ is given in the plane ($\Reals^2$). 
In their range counting queries they consider only a certain type of subcurves called \emph{`inclusion-minimal subcurves'} 
that are the subcurves with smallest length of Fr\'echet distances at most $\delta$ to $Q$. We will simply use this notion to handle \textsc{$(1+\eps)\delta$-ASRS} queries using our generic data structure proposed in \secref{sec:asymDS}. 
As a beneficial application of our result, especially in \thmref{thm:DS}, one can construct a linear-size data structure that efficiently answers the extended queries for polygonal query curves. Moreover, our data structure can report the subtrajectories within an approximation factor of $(1+\eps)$ in any constant dimension. 
Below we recall a lemma from de Berg et al.~\cite{bcg-ffq-13} that also works for a polygonal curve $Q$ of arbitrary edge length, and not necessarily just for a single line segment. 
\begin{lemma}[Lemmas 2 and 3 in \cite{bcg-ffq-13}] \label{lem:inclusionminimal}
	The following statements are true:
	(1) if there exists a subcurve $P'\subseteq P$ with $\delta_F(P',Q)\leq \delta$ then there exists an inclusion-minimal subcurve $P''\subseteq P'$ such that $\delta_F(P'',Q)\leq \delta$, and (2) all inclusion-minimal subcurves of $P$ are pairwise disjoint.
	
\end{lemma}

\begin{proof}
	(1) Let $\I = \cup_{i=1}^{k} B(q_i,\eps)$. If $I\neq \emptyset$ then clearly $P'$ has to intersect $\I$ and the shortest subcurve $P''\subseteq P'$ is a single point of $P \cap \I$ whose Fr\'echet distance to $Q$ is at most $\eps$. Now suppose $\I=\emptyset$ and let $x$ and $y$ be the last and first points along $P'$ that leave and enter the balls $B(q_1,\eps)$ and $B(q_m,\eps)$, respectively. Then $P''$ is the subcurve of $P'$ starting from $x$ and ending at $y$, and having Fr\'echet distance at most $\eps$ to $Q$.
	
	\noindent(2) For the sake of contradiction assume that there are two inclusion-minimal subcurves $P_1$ and $P_2$ that are overlapping along $P$. And by definition we have $\delta_F(P_1,Q)\leq \eps$ and $\delta_F(P_2,Q)\leq \eps$. Observe that $\delta_F(P^*,Q)\leq \eps$ where $P^*=P_1\cap P_2$ and $P^*$ is another subcurve shorter than both $P_1$ and $P_2$. Therefore $P^*$ is an inclusion-minimal subcurve shorter than $P_1$ and $P_2$. This implies that $P_1$ and $P_2$ are no longer inclusion-minimals and therefore this is a contradiction. 
\end{proof}

\noindent{\bf The Algorithm:} The way we exploit the data structure in \thmref{thm:DS} to approximately handle the \textsc{ASRS} queries is as follows: suppose a curve $P$ of size $n$ is given. For every grid path $\C=\segment{c_1\cdots,c_{k}}$ we compute the free space diagram between $P$ and $\C$. This diagram has the domain of $[1,n]\times[1,k]$ and it consists of $(n-1)\times (k-1)$ cells, where each point $(s,t)$ in the diagram corresponds to two points $P(s)$ and $\C(t)$. Given a real $\delta>0$, a point $(s,t)$ in the free space is called \emph{free} if $\|P(s)-\C(t)\|\leq  \delta$ and \emph{blocked}, otherwise. The union of all free points is referred to as the \emph{free space}. The Fr\'echet distance is at most $\delta$ if there exists a monotone path across the free space (see \cite{ag-cfdb-95} for further details on free space diagram). First realize that if there is a subcurve whose Fr\'echet distance is small to $\C$ then there exists an inclusion-minimal subcurve as well by~\lemref{lem:inclusionminimal} (if $\C$ is very short then the inclusion-minimal subcurve would be a single point which is consistent with what we aim for). Without loss of generality assume that $\C$ is aligned along the horizontal axis of the free space. We compute the inclusion-minimal subcurves along $P$ with respect to $\C$ by staying on $c_1$ in free space and going upward until we hit a blocked space, i.e., a point. 
We mark it as the starting point and then we use the classical Alt and Godau's dynamic programming algorithm~\cite{ag-cfdb-95} to reach $c_{k}$ of the lowest free point through a monotone path. This can be done by propagating the rightmost reachable cells of the free space. Since by~\lemref{lem:inclusionminimal} all inclusion-minimal subcurves are disjoint, this entire process takes $O(nk)$ time by repeatedly starting the process from $c_1$ to find inclusion-minimal subcurve. We store the subcurves into a bucket with id associated with $\C$'s id. The index we use for each subcurve to store in the bucket is simply the concatenation of its starting and ending points along $P$. In the query algorithm, we only need to retrieve the inclusion-minimal subcurves stored in the hashtable that are close enough to the rounded query grid path. 

Given that the number of cells presented in \secref{sec:asymDS} is $\N=\max\big\{\big(\frac{1}{\eps}\big)^d, \big(\frac{\D'}{\eps^2}\big)^d\big\}$, where $d$ is a constant, and there are $O(2^{dk}\N^k)$ grid paths of length $k$ and $O(n)$ inclusion-minimal subcurves to store per grid path $\C$ (due to their disjointness property), the space required is:

\begin{center}
	$O\big(n2^{kd}\N^{k}\big)= n \cdot O\big(\max\big\{\big(\frac{1}{\eps}\big)^{kd}, \big(\frac{\D}{\eps^2}\big)^{kd}\big\}\big)$, 
\end{center}

since $\D'\leq \D $ by \lemref{lem:diam}.
The preprocessing is only computing the free space between $P$ and $\C$ in $O(nk) $ time and then the bottom up traversal of the free space along $P$'s axis of it, therefore it takes:
$nk\cdot O\big(\max\big\{\big(\frac{1}{\eps}\big)^{kd}, \big(\frac{\D}{\eps^2}\big)^{kd}\big\}\big)$ time.
We have the following theorem:

\begin{theorem} \label{thm:ASRS}
	Let $P$ be a curve with $n$ vertices in $\Reals^d$, $\delta > 0$ be a real value. For any sufficiently small $\eps>0$, one can construct a data structure of size $n \cdot O\big(\max\big\{\big(\frac{1}{\eps}\big)^{kd}, \big(\frac{\D}{\eps^2}\big)^{kd}\big\}\big)$ and construction time $nk \cdot O\big(\max\big\{\big(\frac{1}{\eps}\big)^{kd}, \big(\frac{\D}{\eps^2}\big)^{kd}\big\}\big)$ such that for any polygonal query curve $Q$ of size $k$ it computes the \textsc{$(1+\eps)\delta$-ASRS} under the Fr\'echet distance in $O(k)$ query time, where $d$ is a constant.  
\end{theorem}

\section{Approximate Time-Window Queries for Spatial Density Maps} \label{sec:timewindow}
In these type of queries,
we assume that all points are contained inside of a polygononal subdivision (map) $M$.
Suppose we are given a set of subdivisions (regions) $\R = \{r_1, \cdots, r_m\}$ of $M$ and a time-stamped point set $\S = \{ (s_1, t_1) , (s_2,t_2), \cdots, (s_n, t_n)\}$ where each point $s_i \in \Reals^d$ appears at time $t_i \in \Reals$ on $M$, and an integer $\theta>0$. W.l.o.g. assume that, for all $1\leq i\leq n$, $0\leq t_i<1$. The aim is to preprocess $\S$ and $\R$ into a data structure so that for any query time-window $W = [q_1,q_2]$ with $q_1 <q_2$, return those subdivisions in $\R$ that contain at least $\theta$ points in $\S$ whose times are within $W$.  
We first start with the preprocessing algorithm below:

\subsection{Preprocessing algorithm}

We describe our preprocessing algorithm below. The input parameters for the algorithm are $\S$, $
\R$, $\theta$, and $\eps$:

\noindent {\bf Step (1):} Compute the smallest and largest times over all samples, i.e., $t_{min}= \min_{1 \leq i \leq n} t_i$ and $t_{max}= \max_{1 \leq i \leq n} t_i$.\\
\noindent {\bf Step (2):} Consider $\I = [t_{max},~2t_{max}-t_{min}]$ and shift the times of all points in $\S$ by $t_{max}$ to fall into the range of $\I$, i.e., $\overrightarrow{t_i} = t_i+t_{max}$, for all $1\leq i\leq n$. \\
\noindent {\bf Step (3):} Subdivide $\I$ into a set of subintervals each of length $\ell = \eps \cdot t_{max}$ and denote the set of the endpoints of the subdivided intervals by $C$. 
\\
\noindent  {\bf Step (4):} For every pair of candidate endpoints $c_1,c_2 \in C$ with $c_1<c_2$, set $W' = [c_1, c_2]$. \\
\noindent  {\bf Step (5):} Associate a bucket in the hashtable to $W'$. Store those regions in $\R$ into the bucket whose number of points is at least $\theta$ fulfilling the time window $W'$.

\begin{lemma} \label{lem:TWsizeC}
	$|C|=O(1/\eps)$. 
\end{lemma}
\begin{proof}
	The number of subintervals induced by the endpoints in $C$ is  clearly:

	$$|C| = \frac{|\I|}{\eps t_{max}} = \frac{2t_{max}-t_{min}- t_{max}}{\eps t_{max}} = \frac{t_{max}-t_{min}}{\eps t_{max}} < \frac{t_{max}}{\eps t_{max}} = O\Big(\frac{1}{\eps}\Big).$$
\end{proof}

\begin{lemma} \label{lem:TWspace_prep}
	The space and construction time of the data structure are $O(m/\eps^2)$ and $O((n+m)/\eps^2)$, respectively. 
	
\end{lemma}
\begin{proof}
	First realize that the space of the data structure is dominated by the size of the hashtable. On the other hand the size of the hashtable is the number of buckets times the number of regions stored into each bucket. Following Lemma~\ref{lem:TWsizeC}, $|C|= O(1/\eps)$ and since the number of buckets is equal to the number of candidate time window $W'= [c_1, c_2]$, and there are $O(|C|^2)$ combinations of $c_1$ and $c_2$ inducing $W'$, there are $O(1/ \eps^2)$ buckets in total. Each bucket can consist of at most $m$ regions in $\R$, therefore the size of the data structure is $O(m/\eps^2)$. 
	
	The construction time is clearly dominated by linearly processing of regions and checking whether each region contains at least $\theta$ points whose times are within $W'$ or not. This takes $O(n+m)$ time, therefore, it would take $O((n+m)/\eps^2)$ overall. 
\end{proof}

\subsection{Query algorithm}

Given a query time-window $W = [q_1, q_2]$, shift $W$ to $\overrightarrow{W} = \big[\overrightarrow{q_1}, \overrightarrow{q_2}\big]$, where $\overrightarrow{q_1} = q_1+t_{max}$ and $\overrightarrow{q_2} = q_2+t_{max}$. Round up $\overrightarrow{q_1}$ and $\overrightarrow{q_2}$ to get  $\overrightarrow{q_1}^+  = \lceil a \rceil$ and $ \overrightarrow{q_2}^+ = \lceil b \rceil$. Similarly, round down $\overrightarrow{q_1}$ and $\overrightarrow{q_2}$ to get $ \overrightarrow{q_1}^- = \lfloor  q^\rightarrow_1 \rfloor$ and $ \overrightarrow{q_2}^- = \lfloor q^\rightarrow_2 \rfloor$. Now consider two time-windows $\overrightarrow{W_1}  = [\overrightarrow{q_1}^+ , \overrightarrow{q_2}^-]$ and $\overrightarrow{W_2}= [\overrightarrow{q_1}^- , \overrightarrow{q_2}^+]$. Retrieve the buckets in the hashtable associated with $\overrightarrow{W_1}$ and $\overrightarrow{W_2}$. Note that there are necessarily two buckets associated with $\overrightarrow{W_1}$ and $\overrightarrow{W_2}$ in the hashtable since $\overrightarrow{q_1}^- ,\overrightarrow{q_1}^+, \overrightarrow{q_2}^-, \overrightarrow{q_2}^+ \in C$. Using the following lemmas we obtain our main theorem in this section:

\begin{lemma} \label{lem:TWquery}
	The query algorithm takes $O(1)$ time.
\end{lemma}

\begin{proof}
	The query algorithm is nothing but a simple shift operation on the endpoints of $W$ which takes $O(1)$ time. Also rounding up and down the resulting endpoints together with retrieving associated buckets from the hashtable take $O(1)$ time overall.
\end{proof}

\begin{lemma} \label{lem:TWapprox}
	For any $0 <\eps < 1/t_{max}-1$, given a query window $W = [q_1, q_2]$, the query algorithm returns two sets of regions $S_1$ and $S_2$ in $\R$ w.r.t. $[(1+\eps)q_1,(1-\eps)q_2]$ and $[(1-\eps)q_1,(1+\eps)q_2]$, respectively, such that $S_1 \subseteq S^*\subseteq S_2$ where $S^*$ is a solution w.r.t. $W$.
\end{lemma}

\begin{proof}
	First note that $S_1$ is the set of regions returned w.r.t. the time-window $\overrightarrow{W_1} = [\overrightarrow{q_1}^+, \overrightarrow{q_2}^-]$ and $S_2$ is the set of regions returned w.r.t. the time-window  $\overrightarrow{W_2} = [\overrightarrow{q_1}^-, \overrightarrow{q_2}^+]$. For a query $W\rightarrow = [a,b]$, it holds that $ \overrightarrow{q_1}^+ - \overrightarrow{q_1}  \leq  \|\overrightarrow{q_1}^+ - \overrightarrow{q_1}^- \|= \eps t_{max}$.  Similarly we have $\overrightarrow{q_1} - \overrightarrow{q_1}^-$, $\overrightarrow{q_2} - \overrightarrow{q_2}^-$, and $\overrightarrow{q_2}^+ - \overrightarrow{q_2}$ at most $\eps t_{max}$.  On other hand, $t_{max}  \leq \overrightarrow{q_1} , \overrightarrow{q_2}$, thus:
	$$\overrightarrow{q_1}^+ \leq \overrightarrow{q_1} + \eps t_{max}= \overrightarrow{q_1} +\eps \overrightarrow{q_1} = (1+\eps)\overrightarrow{q_1} = (1+\eps)(q_1+t_{max})= (1+\eps)q_1,$$
	for any $\eps< 1/t_{max}-1$. 	  
	Using a similar argument, it is not hard to see that $\overrightarrow{q_2}^- \geq (1-\eps)q_2$,  $\overrightarrow{q_1}^- \geq (1-\eps)q_1$, and $\overrightarrow{q_2}^+ \leq (1+\eps)q_2$.   Therefore $\overrightarrow{W_1} = [(1+\eps)q_1,(1-\eps)q_2]$ and $\overrightarrow{W_2} = [(1-\eps)q_1,(1+\eps)q_2]$.  
	
	The optimal set of regions $S^*$ sandwiched between $S_1$ and $S_2$ relies on the fact that for every $\overrightarrow{W_1}$ and $\overrightarrow{W_2}$ with $\overrightarrow{W_1} \subseteq \overrightarrow{W_2}$, $\overrightarrow{W_2}$ contains the regions returned w.r.t. $\overrightarrow{W_1} $ and possibly some more~\cite{bndooh-twdsspd-20}. Therefore $S_1 \subseteq S^* \subseteq S_2$.  
\end{proof}

\begin{theorem} \label{thm:TW}
	For any $0 <\eps < 1/t_{max}-1$, one can build a data structure of size $O(m/\eps^2)$ and construction time $O((n+m)/\eps^2)$ to handle the time-window queries approximately in $O(1)$ query time.
\end{theorem}

	\section{Concluding remarks}
	In this paper we considered  \textsc{ANNS} queries among curves under the Fr\'echet distance and proposed the first result under the continuous Fr\'echet distance in $\Reals^d$.  Our data structure is fully deterministic and simple. In fact our method is a simplification of the ones presented in~\cite{ffk-anncsed-19, dps-sdfsfq-19}. Our data structure can handle both discrete and continuous Fr\'echet distances in asymmetric case and only discrete Fr\'echet in the symmetric case, as well as the \textsc{ASRS} queries more efficiently. In the end, we also looked into the \textsc{TWD} queries studied by Bonerath et al.~\cite{bndooh-twdsspd-20}. We proposed an approximate data structure of linear size and preprocessing time that can approximately handle time-window queries in constant query time.
	An interesting problem to consider, for the future studies, would be to compute the optimization version of  the \textsc{ANNS}, i.e., approximate nearest neighbor problem.
	\bibliographystyle{plainurl}
	\bibliography{datastructure}
	
\end{document}